\documentclass[submission,copyright,creativecommons]{eptcs}

\usepackage{iftex}
\usepackage{prolog}

\ifpdf
  \usepackage[T1]{fontenc}        
\else
  \usepackage{breakurl}           
\fi

\title{Completing the Functional Approach in Object-Oriented Languages}

 \author{Martin Pl\"umicke
 \institute{Duale Hochschule Baden-Württemberg, \\
   Stuttgart, Campus Horb\\
Department of Computer Science,\\
Florianstra{\ss}e 15,\\
D--72160 Horb, Germany}
 \email{pl@dhbw.de}
}
\newcommand{\titlerunning}{Completing the Functional Approach in OO Languages}
\newcommand{\authorrunning}{Martin Pl\"umicke}

\hypersetup{
  bookmarksnumbered,
  pdftitle    = {\titlerunning},
  pdfauthor   = {\authorrunning},
  pdfsubject  = {EPTCS},               
}

\begin{document}
\maketitle

\begin{abstract}
Over the last two decades practically all object-oriented programming
languages have introduced
features that are well-known from functional programming
languages. But many features that were introduced were fragmentary. In \javatx we address
the latter features and propose a completion.
\javatx (i.e. \textbf{T}ype e\textbf{X}tended) is a language based on \java. The predominant new features are global type inference and real function types for lambda expressions. Global type inference means that all type annotations can be omitted, and the compiler infers them without losing the static type property.

We introduce the function types in a similar fashion as in \scala but
additionally integrated them into the \java target-typing as proposed in the so-called strawman approach.

In this paper, we provide an integrated presentation of all \javatx
features. The focus is therby on the automatic inference of type parameters for
classes and their methods, and on the heterogeneous translation of function
types, which permits
the preservation of the argument and return types in bytecode.
\end{abstract}

\section{Introduction}
In this paper we will complete the introduction of functional features to
object-oriented languages employing the example language \java. The main challenge of
this approach is addressing the fundamental differences in both language families. While
functional programming languages only allow immutable data elements,
object-oriented languages contain mutable states. A direct transfer would lead
to unsound languages. Therefore in \java, as in
all other object-oriented languages, the functional features are not introduced
completely. Some features are even slashed.

We closed this gap that spanned  the last 20 years by designing our language \javatx, which
extends the functional features to mutable states. We completed the
  feature \textit{type inference} by extending the local type inference to
  global type inference and we have been completed the concept of \emph{lambda
  expressions} such that functions are treated as first-class
  citizens.

Let us consider the features step by step:
Since version 1.5, the programming language \java has been extended by
incorporating many features from functional programming languages. 
Version 1.5 saw the introduction of generics. Generics are known as parametric polymorphism
in functional programming languages. For parametric polymorphism, the main
challenge was to realise convariance and contravariance, respectively. Whereas in \pizza
(the first approach to the introduction of  functional features to \java \cite{OW97,ORW00})
only invariance was allowed, in \javafive the so-called wildcards were
introduced (a restricted form of existential types) \cite{TEHABG04}
which allow use-side variance, which however then created another problem during method
calling. This was solved  by the capture conversion where each appearance of
a wildcard is substituted by a fresh type variable. By this method, then, the
challenge  for parametric
polymorphism in \java was solved and the features have been completely introduced into
\java. 

In contrast, type inference was only  introduced as local type inference \cite{PT00,OZZ01}
which menas that only some type annotations can be left out. In particular, no
type inference  is possible for recursive lambda expressions and method headers
in method declarations.
Local type inference \cite{PT00,OZZ01} was introduced in the versions 5, 7, and 10. In Java 5, the
automatic determination of parameter instance was introduced. In Java 7, the
diamond operator was introduced. In \javaten, finally, the var keyword for types of local
variables was introduced.

Up until now the feature of global type inference (no type declarations are
necessary, including method arguments and return types, without losing static
typing property) has not been addressed in any object-oriented language.
We solved this problem in \javatx.

Lambda--expressions were introduced in \javaeight. In the so-called project
\emph{Lambda} of Oracle, an extensive discussion over several years yielded the
result. Although, lambda expressions were indeed introduced in \javaeight
function types were not. Instead, there were functional interfaces as target types of
lambda expressions. Many disadvantages can be found due to the lack of function
types \cite{plue17_2}. \javatx counters these disadvantages by introducing function
types in a similar way as in Scala without losing the convenience of functional
interfaces as target type of lambda expressions \cite{plue17_2}. \javatx implements
the so-called strawman approach, which was theoretically given in
\cite{MR09,lambda-spec}. 

The language \javatx corresponds to \java in version 8. Apart from some trivia, we reduced the
language currently by two essential features, exceptions and generics bound by non type-variable 
types (only type variables as bounds are allowed). Furthermore, basic types
(\texttt{int}, \texttt{float}, \texttt{bool}, \ldots) were left out such that the boxed
variants had to be used. But the literals \texttt{1}, \texttt{2}, \texttt{3},
\ldots, \texttt{true}, \texttt{false} were still allowed.
All type annotations were optional.
They could be inferred by our type
inference algorithm (cp. Section \ref{sec:globaltypeinference}).  

The \javatx type system corresponds substantially to the original \javaeight
type system as given in \cite{GoJoStBrBu14}. We have extended the type system
through the introduction of real function types \cite{plue17_2}.



The most significant contribution of this paper lies in the integrated \javatx
presentation of all features 
which complete the functional approach in OO-languages. Furthermore, with this
work we
contribute the third step of the type inference algorithm, the
generated generics (Section \ref{sec:generatedGenerics}) and the heterogeneous
translation of the function types (Section
\ref{sec:heterogenousTranslation}). 

The paper is organised as follows: In Section 2, we shall consider the type
inference algorithm with its two steps: constraint generation and
unification. In Section 3 we shall go on to contribute the generalised type variables
(generated generics). Then in 
Section 4, we shall demonstrate the extended overloading mechanism of
\javatx. The concept of function types with heterogeneous translation in
\javatx 
will be considered  in the fifth section. 
Finally,
we shall conclude with an overview of related work, a summary and an outlook.

\section{Global type inference}

\label{sec:globaltypeinference}
Global type inference guarantees that no type annotations are required. As in
functional programming languages like \haskell, the compiler similarly determines a
principal typing, such that \javatx is statically typed as original \java.
Let us first consider the iterative implementation of the factorial function as
a simple example. 

\begin{example} 
\label{exa:faculty}
Let the class {\texttt{Fac}} 
be given:
\begin{lstlisting}[basicstyle=\tt]
import java.lang.Integer;

class Fac {
    getFac(<@\red{n}@>){
        var res = 1;
        var i = 1;
        while(i<=n) {
            res = res * i;
            i++;
        }
        return res; }
}
\end{lstlisting}
This is a simple example, where the return
and the argument type of \langelem{getFac} are left out. The type inference
algorithm has to infer the types.
\end{example}


    

The type inference algorithm consists of three steps: \emph{constraint generation},
\emph{type unification}, and \emph{generate generics}.

In this section, we briefly summarise the first two steps, refering
to \cite{plue15_2,plue09_1,stpl18} for details. One of the key aspects of this
paper is the third step 
\emph{generalized type variables} which
will be considered in the next section.

In the following, we use two symbols for subtyping: One the one hand \olsub
stands for the subtype relation of two types and on the other hand $\lessdot$
means that two types should be unified such that they become subtypes. This
means, $ty \lessdot ty'$ is a type constraint, where $ty$ and $ty'$ are type expressions
consisting of (generic) class types and type variables.

\subsection{Constraint generation}
\label{sec:constraintgeneration}
In a traversing of the code, to each node of the methods’ statements and
expressions is assigned a type. If the corresponding types are left out, a
fresh type variable
is mapped as type placeholder; otherwise, the known type is
mapped.

  During the traversing, a set of type constraints $\set{\ol{ty \lessdot ty'}}$
  is generated. The constraints represent the type conditions as defined in
  the \java specification \cite{GoJoStBrBu14}. 
  For more details, see the function \textbf{TYPE} in \cite{plue15_2}.
  Type variables \texttt{A}, \texttt{B}, \texttt{C}, \ldots are written in
  upper-case letters in typewriter font.

\begin{figure}[htbp]
  \centering
\begin{minipage}[c]{0.55\linewidth}
  \begin{lstlisting}
class Fac {
  <@\tvar{N}@>getFac(<@\tvar{O}@>n) {
    <@\tvar{P}@>res = 1;
    <@\tvar{R}@>i = 1;
     while((i<@::\tvar{R})@> <= (n<@::\tvar{O})@>)<@::\tvar{T}@> {
        (res<@::\tvar{P})@>=((res<@::\tvar{P})@>*(i<@::\tvar{R})@>)<@::\tvar{U}\!\!\!@>;
        (i<@::\tvar{R})@>++;
      }
      return(res<@::\tvar{P})@>;
  }
}
  \end{lstlisting}
\end{minipage}
\ \ 
\begin{minipage}[c]{0.35\linewidth}
$\set{(\tvar{P} \lessdot  \tvar{N}),\\
\hspace*{0.2cm}  
(\tvar{U} \lessdot \tvar{P}),\\
\hspace*{0.2cm}  
(\tvar{O} \lessdot \mathtt{java.lang.Number}),\\
\hspace*{0.2cm}  
(\tvar{R} \lessdot \mathtt{java.lang.Number}), \\
 \hspace*{0.2cm}  
(\mathtt{java.lang.Boolean} \doteq \tvar{T}),\\
\hspace*{0.2cm} 
(\mathtt{java.lang.Integer} \doteq \tvar{U}), \\
 \hspace*{0.2cm}
(\tvar{R} \lessdot \mathtt{java.lang.Integer}), \\
 \hspace*{0.2cm} 
(\tvar{P} \lessdot \mathtt{java.lang.Integer})}$
  \end{minipage}  
  \caption{Result of the constraint generation}
  \label{fig:constraintgeneration}
\end{figure}

  \begin{example} 
The result of the constraint generation of the \langelem{Fac}-example is presented
in Fig.~\ref{fig:constraintgeneration}.
On the left-hand side we present the essential type variables of the \langelem{getFac}--method
which are mapped  to nodes. A type variable \langelem{O} 
assigned to a variable \langelem{n} in the program is written as
"\langelem{(n::}\langelem{O)}". On the right-hand
side, the    
generated constraints are presented.

The search space is restricted by the imported types with their sub- and
  supertypes and packages,
  respectively. In this example \langelem{java.lang.Integer} with its supertypes.
\end{example}

\subsection{Unification}
For the set of type constraints $\set{\ol{ty \lessdot
      ty'}}$, general unifiers (substitutions) $\sigma$
  are demanded, such that $\ol{\exp{\sigma}{ty} \olsub \exp{\sigma}{ty'}}$.
The result of the type unification is a set of pairs
$(\set{\ol{(T \lessdot T')}}, \sigma)$,
where $\set{\ol{(T \lessdot T')}}$ is a set of remaining constraints consisting of
  two type variables and $\sigma$ is a general unifier. 

The type unification algorithm is given in \cite{plue09_1,stpl18}. There
we proved that the unification is indeed not unitary, but finitary,
meaning that there are finitely most general unifiers.

Let us consider the application of the type inference algorithm to the factorial
example (Example \ref{exa:faculty}):
\begin{figure}[htbp]
  \centering
\begin{minipage}[c]{0.35\linewidth}
$\sarray{@{}l@{}l}{
\set{(\emptyset, [& (\tvar{U} \mapsto \mathtt{java.lang.Integer}),\\
& (\tvar{P} \mapsto \mathtt{java.lang.Integer}),\\
&
(\tvar{R} \mapsto \mathtt{java.lang.Integer}),\\
& (\tvar{O} \mapsto \mathtt{java.lang.Integer}),\\
&
(\tvar{N} \mapsto \mathtt{java.lang.Integer}),\\
& (\tvar{T} \mapsto
\mathtt{java.lang.Boolean})]}}$
\end{minipage}
\ \ 
\begin{minipage}[c]{0.60\linewidth}
\begin{lstlisting}
class Fac {
  Integer getFac(Integer n) {
    Integer res = 1;
    Integer i = 1;
    while((i<@::\texttt{Integer})@> <= (n<@::\texttt{Integer})@>)<@::\texttt{Boolean}\,@>{
      (res<@::\texttt{Integer})@> = 
        ((res<@::\texttt{Integer})@> * (i<@::\texttt{Integer})@>)<@::\texttt{Integer}@>;
      (i<@::\texttt{Integer}@>)++;}
    return(res<@::\texttt{Integer})@>;}}
\end{lstlisting}
\end{minipage}
  \caption{Result of the type unification}
  \label{fig:unificationresult}
\end{figure}
\begin{example} 
The result of type unification for the \texttt{Fac}-example is illustrated in
Fig.~\ref{fig:unificationresult} on the left-hand\nolinebreak[2] side.
In this example, no constraints remain that only consist of type variables.
Furthermore, there is only one general unifier.
On the right-hand side, we instantiated the type variables by the determined types.
\end{example}
\section{Generalised type variables}
\label{sec:generatedGenerics}
 In the previous section we considered the unification of type expressions. In this section,
we shall consider the remaining constraints. In the existing type
inference algorithm of functional programming languages without
subtyping (e.g. \haskell or \sml), the remaining type variables are
generalised such that any type can be instantiated if the function
is used.

Following this idea, the remaining type variables become (bound)
type parameters of the class and its methods, respectively, where the
left-hand side of a constraint is a type parameter and the right-hand
side is its bound.
%
%
%
%
Additionally, due to the \java restrictions of type parameters, some type
parameters have to be collected to one new type parameter. 
%
\begin{figure}[htbp]
\begin{tabular}{p{\textwidth/2}@{}p{\textwidth/2}}
\begin{lstlisting}
class TPHsToGenerics<@\footnotemark@> {
    id = x -> x;

    id2 (x) {
        return id.apply(x);}
    
    m(a, b){ return b; }
    
    m2(a, b) {
        var c = m(a,b);
        return a; }
}
\end{lstlisting}\pause
&
\begin{lstlisting}
class TPHsToGenerics {
    <@\red{Fun1\$\$<UD, ETX>\footnotemark}@> id = (<@\red{DZP}@> x) -> x;

    <@\red{ETX}@> id2(<@\red{V}@> x) {
        return id.apply(x);}
    
    <@\red{AI}@> m(<@\red{AM}@> a, <@\red{AN}@> b){ return b; } 
    
    <@\red{AA}@> m2(<@\red{AB}@> a, <@\red{AD}@> b){
        <@\red{AE}@> c = m(a,b);
        return a; }
}
\end{lstlisting}
  \end{tabular}
\vspace*{-0.5cm}
  \caption{Class \texttt{TPHsToGenerics} before and after tree traversing}
  \label{fig:TPHsToGenerics}
\end{figure}
Let us start with a motivating example:

\begin{example}
  \label{exa:FoGG}
  On the left-hand side in Fig.~\ref{fig:TPHsToGenerics}, a \javatx program is given. The identity
  function is mapped to the field \langelem{id}. In the method \langelem{id2} the
  identity function is called. In the method \langelem{m2} the method \langelem{m}
  is called. 

  The application of the traversing step is presented on the right-hand side
  where we leave out inner type variables. The result of the type unification
  is $\set{(cs, [\,])}$ where the remaining set of constraints is:  
{\centering $cs = \sarray{l@{}l}{\set{&\mathtt{UD \lessdot DZP,\ DZP
  \lessdot ETX
},\ \mathtt{V \lessdot UD,\ AN \lessdot AI,\ AB \lessdot AA,\ AB \lessdot AM,\ AD \lessdot AN,\  AI \lessdot
      AE}
  }}$
}
\end{example}
\subsection{Family of generated generics}

\addtocounter{footnote}{-1}
\footnotetext{\texttt{TPHsToGenerics} means transform type placeholders from
  the type inference result to
  generic type variables of \java classes and methods, respectively.}
\addtocounter{footnote}{1}
\footnotetext{$\funtypestartt{\textit{N}}{.}{.}$ is a function type introduced in Section \ref{sec:realfuntiontypes}}

We divide up  the set of remaining constraints $cs$ by transferring it to
bound type variables of the class and each method of the class, respectively.
Thus we build a family of generated generics $FGG$
where the index set is given as the class name and its method names.

\begin{definition}[Family of generated generics]
  \label{def:FoGG}
The \emph{family of  generated generics} is defined as 
$$FGG = (FGG_{in})_{in \in \textit{CLM}},$$
where $\textit{CLM} =\set{cl} 
  \cup \set{m \, | \, m \textrm{ is method in } cl}$ 
is the index set of the class name and its methods' names.

Let $cs$ be a set of remaining constraints as result of the type
unification. $cs$ is transferred  to the family of generated generics $FGG$ where the set of generated generics of
the class $FGG_{cl}$ are given as:  
\begin{itemize}
\item All type variables of the fields with its bounds, including the initialisers,

\item the closure of all bounds of type variables of the fields with its
  bounds, and

\item all unbound type variables of the fields and all unbound
bounds. They get \langelem{Object} as bound.
\end{itemize}
The set of generated generics $FGG_{m}$ of its methods $m$: are given, respectively, as:
\begin{itemize}
\item The type variables of the method \red{$m$} with its bounds, where the bounds are also
  type variables of the method,

\item all type variables of the method m with its bounds, where the
  bounds are type variables of fields, and

\item all unbound type variables of the method
    $\mathtt{\red{m}}$ and all type variables of
    $\mathtt{\red{m}}$ which bounds are not type variables of
    $\mathtt{\red{m}}$ with \langelem{Object} as bound.
\end{itemize}
\end{definition}

\begin{figure}[htbp]
  \centering
\begin{lstlisting}
class TPHsToGenerics <@\red{<UD\ extends}\ \greenA{DZP,\ DZP\ extends}\
\red{ETX>}@> {

    <@\red{Fun1\$\$<UD, ETX>}@> id = x -> x;

    <@\red{<V extends UD> ETX}@> id2(<@\red{V}@> x) { return id.apply(x); }
    
    <@\red{<AM, AN extends AI, AI> AI}@> m(<@\red{AM}@> a, <@\red{AN}@> b){ return b; }
    
    <@\red{<AA, AB extends AA, AD, AE>}@> <@\red{AA}@> m2(<@\red{AB}@> a, <@\red{AD}@> b){
        <@\red{AE}@> c = m(a,b);
        return a;}
}
\end{lstlisting} 
  \caption{Generated generics of the class \texttt{TPHsTo\-Ge\-ner\-ics}}
  \label{fig:FOGG}
\end{figure}
\begin{example}
The family of generated generics for the class \langelem{\texttt{TPHs\-To\-Generics}}
from Example \ref{exa:FoGG} is given as:


\medskip
\textbf{The set of generated generics  $FGG_{\mathtt{TPHsToGenerics}}$ of the class:}
\begin{itemize}
\item Type variables of the fields with its bounds:
   {\red{$\set{\mathtt{UD\,\lessdot\, DZP}}$}}.

\item Closure of all bounds of type variables of the fields with its bounds:
    {\red{$\set{\mathtt{DZP\,\lessdot\,ETX}}$}}.

\item All unbound type variables of the fields and all unbound bounds:
   {\red{$\set{\mathtt{ETX\,\lessdot\,Object}}$}}.
\end{itemize}

\pagebreak[2]
\textbf{The set of generated generics $FGG_{\mathtt{id2}}$:}
\begin{itemize}
\item All type variables of the method m with its bounds, where the
  bounds are type variables of fields
{$\set{\red{\mathtt{V}} \lessdot \red{\mathtt{UD}}}$}.
\end{itemize}
\textbf{The set of generated generics $FGG_{\mathtt{m}}$:}
\begin{itemize}
\item The type variables of the method $\red{m}$ with its bounds,
    where the bounds are also type variables of the method:
{$\set{\red{\mathtt{AN \lessdot AI }}}$}.
\item All unbound type variables of the method $\red{m}$ with
  \langelem{Object} as bound:
 {$\set{\red{\mathtt{AM \lessdot Object,\  AI \lessdot Object}}}$}.
\end{itemize}

\textbf{The set of generated generics $FGG_{\mathtt{m2}}$:}
 \begin{itemize}
   \item The type variables of the method $\red{m}$ with its bounds,
    where the bounds are also type variables of the method:
{$\set{\red{\mathtt{AB \lessdot AA}}}$}.





  \item All unbound type variables of the method
    $\mathtt{\red{m}}$ and all type variables of
    $\mathtt{\red{m}}$ which bounds are not type variables of
    $\mathtt{\red{m}}$ with \langelem{Object} as bound:  
{$\set{\red{\mathtt{AD \lessdot Object, AE \lessdot Object}}}$}
\end{itemize}


The mapping of the family to the class and its methods in the \javatx program is
presented in Fig. \ref{fig:FOGG}, where the bounds \langelem{Object} are left
out.

\end{example}

This is not yet a correct \java-program. In the method \texttt{m2}, the type
\texttt{AD} of the
second parameter \texttt{b} of the method-call of \texttt{m} must be a subtype
of the type \texttt{AE} of the local variable \texttt{c}, as in the method
\texttt{m}, the argument type
\texttt{AN} is a subtype of the return type \texttt{AI}. We address this problem
by extending the family of generated generics to the completed family of
generated generics.

\begin{definition}[Completed family of generated generics]
  \label{def:cFoGG}
  Let $cs$ be the remaining constraints after unification and $FGG$ be the family of generated generics.
The \emph{completed family of  generated generics} $CFGG$ is defined as 
  \begin{itemize}
  \item $CFGG_{cl} = FGG_{cl}$
\item $CFGG_{m}$ corresponds to $FGG_{m}$ where $\tvar{T_i} \lessdot
  \mathtt{Object}$ is substituted by $\tvar{T} \lessdot \tvar{R}$ for each method
  call
  
  {\centering
   $\mathtt{((rcv::cl).m'(\ol{(e::\tvar{ty})})::\tvar{rty})}$\\}

   in the method \red{$m$} of the class {\normalshape\texttt{cl}} with the signature $\oneargtyp{\ol{ty}'}{rty'}$ where
   $\tvar{T} \in \expsf{TVar}{\tvar{ty_i}}$,  $\tvar{T'} 
   \in \expsf{TVar}{ty'_i}$, $\tvar{R'} \in \expsf{TVar}{rty'}$,
   $\tvar{R} \in \expsf{TV}{\tvar{rty}}$, $T \lessdot T' \lessdoteq R' \lessdot
   R$ is in the transitive closure of $cs$, and 
  $\tvar{T'}  \lessdoteq \tvar{R'}$ is in the transitive closure of $CFGG_{m'}$,
 \end{itemize}
\end{definition}

\begin{example}
  \label{exa:comFoGG}
  In the completed family of generated generics of the class
  \langelem{TPHsToGenerics} in the set of generated generics $FGG_\mathtt{m2}$ the
  bound of \langelem{AD} is
  changed from \langelem{Object} to \langelem{AE}:

   {\centering $\set{\red{\mathtt{AD}} \lessdot \red{\mathtt{AE}}}$\\}

   as the method \langelem{m} is called in \langelem{m2}:

    {\centering $\langelem{\tvar{AE} c = }\exptt{m}{a::\tvar{AB},
        b::\tvar{AD}},$\\}

   $\langelem{AD} \lessdoteq \langelem{AN} \lessdoteq \langelem{AI} \lessdoteq
   \langelem{AE}$ is in the transitive closure of $cs$, and
   $\langelem{AN} \lessdoteq \langelem{AI}$ is in the transitive closure of 
  $FGG_\mathtt{m}$.
\end{example}

\begin{theorem}
   The completed family of generated generics of any class fulfills the typing
   rules of a \java class.
 \end{theorem}
 \begin{proof}   
   The set of constraints which is generated during the traversing step
   (cp. Section \ref{sec:constraintgeneration}) satisfies the \java typing rules
   \cite{plue15_2a}.

   This set is the input of the type unification
   algorithm. The type unification algorithm has been proven as sound
   \cite{plue09_1}, such that the remaining constraints $cs$ of the output of  the
   type unification fulfill the \java typing rules.

   $FGG$ is a partition of the $cs$. As $FGG$
     assigns each type variable to the class and to the methods, respectively,
     where the type variable is declared and furthermore, as the type variables are unique
     (during the traversing step \emph{fresh} type variables are generated), each
     member of  $FGG$ contains all its type variables and the bounds are
     correct.

     However, there is no guarantee that all relationships of type variables of a method
     are contained in the corresponding member of $FGG$ as in the
     last step of Def.~\ref{def:FoGG} the bounds which are not type variables
     of the method are substituted by \texttt{Object}.

     This means that for type variables \texttt{T} and
     \texttt{R} of a method \texttt{m}  it is possible that there are type constraints $\mathtt{T'}
     \lessdoteq \mathtt{R'}$ in $cs$ with $\mathtt{T} \lessdot \mathtt{T'}
     \lessdoteq \mathtt{R'} \lessdot \mathtt{R}$.

     If such type variables $\mathtt{T'}$ and $\mathtt{R'}$ exist, it means
     that in $\mathtt{m}$ a method $\mathtt{m'}$ is called where $\mathtt{T'}$
     is contained in a type of a parameter of $\mathtt{m}'$ and $\mathtt{R'}$ in
     the result type of $\mathtt{m'}$.

     In $CFGG$ these relationships were added such that $CFGG$ fulifills
     the \java typing rules.
\end{proof}

  \textbf{Remark:} Let us consider two aspects of the proof in more detail. If the method
  $\mathtt{m'}$ is overloaded or overridden all cases are included as for each
  overloading/overriding an individual result set is generated by the unification.

  We have proven that the completed family of generated generics $CFGG$ fulfills the
  \java typing rules, but we have not proven that the binary relation induced
  by the $CFGG$ is a correct \java subtyping ordering. This will be addressed in
  Section \ref{sec:javaconform}.

\medskip
The algorithm which completes the family of generated generics visits all
methods along the call-graph and collects the bounds.

\medskip
Here we offer a further more complex example, where two methods are mutual recursive.
\begin{figure}[htbp]
\begin{minipage}[t]{0.45\linewidth}
    \begin{lstlisting}
class Mutual {



  Pair<BB,DD> m1(B x, C y) {
      D y2 = m2(x, y).snd();
      return new Pair<>(id(x),y2);
  }



  Pair<HH,GG> m2(F x, G y) {
      H x2 = m1(x, y).fst();
      return new Pair<>(x2, id(y));
  }


  I id(J x) { return x; }
}
  \end{lstlisting}
\end{minipage}
\ \ \ \ \ \ \ \  
\begin{minipage}[t]{0.4\linewidth}
\begin{lstlisting}
class Mutual  {
    
  <B extends BB, BB, C extends D,
   D extends DD, DD>
  Pair<BB,DD> m1(B x, C y) {
      D y2 = m2(x, y).snd();
      return new Pair<>(id(x),y2);
  }

  <F extends H, H extends HH, HH,
   G extends GG, GG>
  Pair<HH,GG> m2(F x, G y) {
      H x2 = m1(x, y).fst();
      return new Pair<>(x2, id(y));
  }

  <J extends I, I>
  I id(J x) { return x; }
}
\end{lstlisting}
\end{minipage}
  \caption{Further example for the completed family of generated generics}
  \label{fig:mutual}
\end{figure}
\begin{example}
  Let the class \langelem{Mutual} in Fig.~\ref{fig:mutual} be given. On the
  left-hand side the types are inferred. Furthermore, the set of remaining constraints is given as
  
  {\centering $cs = \set{\mathtt{B} \lessdot \mathtt{J},\ 
    \mathtt{BB} \lessdot \mathtt{H},\ 
    \mathtt{B} \lessdot \mathtt{F},\ 
    \mathtt{C} \lessdot \mathtt{G},\ 
    \mathtt{GG} \lessdot \mathtt{D},\ 
    \mathtt{F} \lessdot \mathtt{B},\ 
    \mathtt{G} \lessdot \mathtt{C},\ 
    \mathtt{G} \lessdot \mathtt{J},\ 
    \mathtt{I} \lessdot \mathtt{BB},\ 
    \mathtt{I} \lessdot \mathtt{GG},\ 
    \mathtt{J} \lessdot \mathtt{I}
  }$\\}
\bigskip
  The family of generated generics is given as:

  $FGG_\mathtt{m1} =  \set{\mathtt{B} \lessdot \mathtt{Object},\  \mathtt{C} \lessdot
    \mathtt{Object},\  \mathtt{D} \lessdot \mathtt{DD},\  \mathtt{DD} \lessdot
    \mathtt{Object},\  \mathtt{BB} \lessdot
    \mathtt{Object}}$

    $FGG_\mathtt{m2} =  \set{\mathtt{F} \lessdot \mathtt{Object},\  \mathtt{G}
      \lessdot \mathtt{Object},\  \mathtt{H} \lessdot \mathtt{HH},\  \mathtt{HH}
      \lessdot \mathtt{Object},\  \mathtt{GG} \lessdot
    \mathtt{Object}}$

    $FGG_\mathtt{id} =  \set{\mathtt{J} \lessdot \mathtt{I},\  \mathtt{I}
      \lessdot \mathtt{Object}}$

    This leads to the completed family of generated generics given in the
    program on the right side of Fig.~\ref{fig:mutual}.

\end{example}

In the following section, we shall transform the completed family of generated
generics to \java generics of the class and its methods, respectively. 

\subsection{\java-conforming binary relation of type parameters}
\label{sec:javaconform}
The set of remaining constraints as well as each element of the family of
generated generics are arbitrary binary relations. There are two conditions in
\java which all members of the family of generated generics have to fulfill:
\begin{itemize}
\item The reflexive and transitive closure must be a partial ordering (the
  subtyping relation is a partial ordering). 
\item Two different elements have no infimum (multiple inheritance is prohibited).
\end{itemize}

Now in Fig.~\ref{fig:nonconformCFGG} we offer two examples of completed families of
generated generics induced by the type inference which indeed corresponds to the
\java typing rules, but these are not correct \java classes.
\begin{figure}[htbp]
  \begin{minipage}[t]{0.45\linewidth}
      \begin{lstlisting}
class Cycle {
    <L extends M,M extends L> 
    void m(L x, M y) {
        y = x;
        x = y;
    }
}   
  \end{lstlisting}
  \end{minipage}
  \ \ \ \ \ \ \ 
  \begin{minipage}[t]{0.45\linewidth}
     \begin{lstlisting}
class Infimum {
    <A extends B, A extends C, B, C>
    m(A a, B b, C c) {
        b = a;
        c = a;
    }
}
  \end{lstlisting} 
  \end{minipage}
  \caption{Non \java-conform completed families of generated generics}
  \label{fig:nonconformCFGG}
\end{figure}

In the class \texttt{Cycle} the binary relation of the generics is no partial
ordering as it is not antisymmetric (\texttt{L}\nopagebreak \olsub \texttt{M},
\texttt{M} \olsub \texttt{L}, but \texttt{L} $\neq$ \texttt{M}).
In the class \texttt{Infimum} the generics \texttt{B} and \texttt{C} have the
infimum \texttt{A}.

The general approach  is to equalise type variables  by a surjective map $h$
that preserves the subtype relation such that  for $T \lessdoteq T'$ holds true
$\exp{h}{T} \lessdoteq \exp{h}{T'}$.

\bigskip
Let $CFGG$ be the completed family of generated generics of a class. For any
$C \in CFGG$ the cycles and the infima are removed:

First, the cycles are removed. For any ($\tvar{T} \lessdot \tvar{K} \lessdot \tvar{G}
\lessdot ... \lessdot \tvar{T}$) in $C$:

\begin{Listenullparsep}
\item Define for a fresh type variable \tvar{X}: $\exp{h}{cy} := \tvar{X}$ for
  $cy$ in the cycle.
\item Apply $h$ to any element of any member of $CFGG$.
\end{Listenullparsep}

In a second step, the infima are removed:
    Apply the following steps until there are no infima in $C$:

    For any 
$$\textit{Constr}_{\tvar{T}} = \set{(\tvar{T} \lessdot \tvar{R}),\  (\tvar{T} \lessdot \tvar{S}),\  \ldots } \subseteq C$$

\begin{Listenullparsep}
\item Define for a fresh type variable \tvar{X}: $\exp{h}{\tvar{T}} :=
  \tvar{X}$,  $\exp{h}{\tvar{R}} := \tvar{X}$,  $\exp{h}{\tvar{S}} :=
  \tvar{X}$, \ldots
\item Apply $h$ to any element of any member of $CFGG$.
\end{Listenullparsep}

For all remaining type variables $h$ is defined as identity: $\exp{h}{ty} = ty$.

\begin{lemma}
  \label{lemma:javaconform}
  Let $C$ be a member of the completed family of generated generics of a class
  and $h$ the corresponding surjective map. For $\tvar{T} \lessdoteq \tvar{T}'$ holds true
  $\exp{h}{\tvar{T}} \lessdoteq \exp{h}{\tvar{T}'}$. 
\end{lemma}

\begin{proof}
  For $\tvar{T} \lessdoteq \tvar{T}'$ either $\exp{h}{\tvar{T}} = \exp{h}{\tvar{T}'}$ holds. Then 
  $\exp{h}{\tvar{T}} \lessdot \exp{h}{\tvar{T}'}$ holds true as $\lessdot$ is reflexive.
  In the other case, there are type variables $\tvar{Y}_1 \ldo \tvar{Y}_n$ with  $\tvar{T} \lessdoteq
  \tvar{Y}_1  \lessdoteq \ldots \lessdoteq \tvar{Y}_n \lessdoteq \tvar{T}'$ and $\exp{h}{\tvar{Y}_i} =
  \exp{h}{\tvar{Y}_{i+1}}$ for some $i \in \set{1 \ldo n}$. From this, it follows that
  $\exp{h}{\tvar{T}} \lessdoteq \exp{h}{\tvar{Y}_1}  \lessdoteq \ldots \lessdoteq
  \exp{h}{\tvar{Y}_n} \lessdoteq \exp{h}{\tvar{T}'}$.


\end{proof}

Removing cycles in the class \langelem{Cycle} (Fig.~\ref{fig:nonconformCFGG}) leads to the surjective
mapping $h$ with
  $\exp{h}{\red{\tvar{L}}} = \blue{\tvar{X}}$ and
  $\exp{h}{\red{\tvar{M}}} = \blue{\tvar{X}}$.
Removing infima in the class \langelem{Infimum}
(Fig.~\ref{fig:nonconformCFGG}) leads to the surjective
mapping $h$ with
  $\exp{h}{\red{\tvar{M}}} = \blue{\tvar{X}}$ and
  $\exp{h}{\red{\tvar{N}}} = \blue{\tvar{X}}$.
In Figure \ref{fig:javaconformCFGG} the resulting clases are given.
\begin{figure}[htbp]
  \begin{minipage}[t]{0.45\linewidth}
      \begin{lstlisting}
class Cycle {
    <@\blue{<X>}@> void  m(X x, X y) {
        y = x;
        x = y;
    }
}
  \end{lstlisting}
  \end{minipage}
  \ \ \ \ \ \ \ 
  \begin{minipage}[t]{0.45\linewidth}
     \begin{lstlisting}
class Infimum {
  <@\red{<\blue{X}>}@> void  m(X x, X y, X z) {
        y = x;
        z = x;
    }
}
  \end{lstlisting} 
  \end{minipage}
  \caption{\java-conform completed families of generated generics}
  \label{fig:javaconformCFGG}
\end{figure}

\section{Overloading}
\label{sec:overloading}
The example in Figure \ref{fig:OL} illustrates the extended overloading mechanism of \javatx. 
In the class
\langelem{OL}, the method name \langelem{m} is overloaded by two different method
declarations.
If the types \texttt{Integer},
\texttt{Double}, \texttt{String}, and \texttt{Boolean} are visible, the type of
the first method \langelem{m} is: 
\begin{figure}[htbp]
  \begin{minipage}[t]{0.45\linewidth}
  \begin{lstlisting}
  class OL {
    <@\blue{m(x) \{ return x + x; \}}@>
	
    <@\bluered{m(x) \{ return x || x; \}}@>
  }   
 \end{lstlisting}
  \end{minipage}
  \ \ \ \ \ \ \ 
  \begin{minipage}[t]{0.45\linewidth}
     \begin{lstlisting}	   
  class OLMain {
    <@\red{main(x)}@> { 
      var ol = new OL(); 
      return ol.m(x); } 
  }  
\end{lstlisting}
 \end{minipage}
  \caption{Example für the \javatx overloading mechanism}
  \label{fig:OL}
\end{figure}

\smallskip
{\centering $\sarray{rcl}{
   \blue{\mathtt{OL.m}} & \blue{:}
  & \blue{\mathtt{Integer} \rightarrow \mathtt{Integer}} 
   \ \ \blue{\&} \ \ 
   \blue{\mathtt{Double} \rightarrow \mathtt{Double}} 
    \ \ \blue{\&} \ \ 
    \blue{\mathtt{String} \rightarrow \mathtt{String}},}$\\}

\smallskip
  as \langelem{+} is an overloaded operation.
  The second method \langelem{m} has the type:

\smallskip
{\centering $\sarray{rcl}{
\bluered{\mathtt{OL.m}} & \bluered{:}
   & \bluered{\mathtt{Boolean} \rightarrow \mathtt{Boolean}}.}$\\}

\smallskip
In the class \langelem{OLMain}, an instance of \langelem{OL} is created and on the
instance the overloaded method \langelem{m} is called. 
This means that the method
declaration \langelem{main} is overloaded by the types of both declarations:

\smallskip
{\centering $\sarray{rcl}{
  \red{\mathtt{OLMain.main}} & \red{:}
  & \red{\mathtt{Integer} \rightarrow \mathtt{Integer}} 
   \red{\ \&\ } 
   \red{\mathtt{Double} \rightarrow \mathtt{Double}} 
    \red{\ \&\ } 
    \red{\mathtt{String} \rightarrow \mathtt{String}}\\
   &
   \red{\&} & \red{\mathtt{Boolean} \rightarrow \mathtt{Boolean}.} 
}$\\}

\smallskip
This example shows the extended overloading mechanism in \javatx.
Standard-\java only allows the overloading of method identifiers, 
meaning that multiple method declarations with the same identifier have to be
declared. In contrast, in \javatx, one declaration could be overloaded,
which means that one declaration has different types.

\bigskip

\section{Function types}
\label{sec:realfuntiontypes}
In \javaeight, lambda expression were indeed introduced, but not function types. 
Instead, there are functional interfaces as target
types of lambda expressions. There are many disadvantages due to
of the lack of function types (cp. \cite{plue17_2}). \javatx counters these disadvantages
by introducing function types in a similar way as in \scala without losing
the convenience of functional interfaces as target type of lambda
expressions \cite{plue17_2}. \javatx implements the so-called
strawman approach, which was theoretically given in \cite{lambda-spec, MR09}.
We extended in \javatx two sets of special functional interfaces with
declaration-site variance type parameters:

\begin{minipage}[t]{0.48\linewidth}
\begin{lstlisting}[numbers=none,basicstyle=\tt\footnotesize]]
 interface Fun<@\itshape{N}@>$$<-T1, ..., -T<@\itshape{N}@>,+R> {
   R apply(T1 arg1, ... , T<@\itshape{N}@> arg<@\itshape{N}@>); 
 }
\end{lstlisting}
\end{minipage}
\begin{minipage}[t]{0.48\linewidth}
\begin{lstlisting}[numbers=none,basicstyle=\tt\footnotesize]
   interface FunVoid<@\itshape{N}@>$$<-T1, ..., -T<@\itshape{N}@>> {
     void apply(T1 arg1, ... , T<@\itshape{N}@> arg<@\itshape{N}@>);
   }
\end{lstlisting}
\end{minipage}

where

   \begin{Listenullparsep}
   \item $\funtypestartt{{\itshape N}}{T'_1 ... T'_\mathit{N}}{T_0} \olsub \funtypestartt{{\itshape N}}{T_1  ...
     T_\mathit{N}}{T'_0}$ iff $\mathtt{T_i} \olsub \mathtt{T'_i}$ 
 \item In
   \texttt{Fun{\itshape N}\$\$} no wildcards are allowed.
   \end{Listenullparsep}

  \red{The Lambda-expressions are explicitly typed by
    \texttt{Fun{\itshape N}}\$\$-types}.

\subsection{Heterogenous translation}
\label{sec:heterogenousTranslation}
One of the main arguments against introducing real function types into \java
had been the type erasure during translation to byte-code
\cite{lambda-stateV5}. This means that any function type
$\funtypestartt{{\itshape N}}{T'_1 ... T'_\mathit{N}}{T_0}$ is translated to
\texttt{Fun{\itshape N}\$\$} without argument and without result
types. Therefore, in \javatx, the function types are translated
heterogeneously, meaning that during translation no type erasure (deleting the arguments of
generic types) is carried out (cp \cite{ORW00} for heterogeneous and
homogenous translation).

\medskip

In order to gain an understanding of the problem, let us consider the following example:

\begin{example}
  \label{exa:OLFun}
  Let the following \javatx program be given:
  
{\normalshape
\begin{lstlisting}
  class OLFun {
    m(f) {
      var x;
      x = f.apply(x+x);
      return x;}
   }
  \end{lstlisting}
}
The inferred method headers (cp. Section \ref{sec:overloading}) and the corresponding descriptors in bytecode 
are:
{\normalshape
\begin{lstlisting}[basicstyle=\ttfamily\small]
Double m(Fun1$$<Double,Double>);    <@\red{descriptor:(LFun1\$\$;)Double;}@>
Integer m(Fun1$$<Integer,Integer>); <@\red{descriptor:(LFun1\$\$;)Integer;}@>
String m(Fun1$$<String,String>);    <@\red{descriptor:(LFun1\$\$;)String;}@>
\end{lstlisting}
}
Now method resolving, in \textsf{JVM} done by the descriptors, is no longer possible as all three methods have the
same argument \langelem{Fun1\$\$}.

\end{example}
This problem could be solved by heterogenous translations which preserve the
arguments in the descriptors. 
Therefore, following \cite{ORW00}, we translate a type
  $\funtypestar{\itshape{N}}{ty_1 \ldo ty_n}{ty_0}$ 
  to a string
  
 {\centering  
  \verb;Fun;\texttt{\textit{N}}\verb;$$$_$;$\tilde{ty_1}$\verb;$_$; \ldots
  \verb;$_$;$\tilde{ty_n}$\verb;$_$;$\tilde{ty_0}$\verb;$_$;, \\}
where \emph{$\tilde{ty_i}$} are the translations of the type parameters $ty_i$.
 The translations are subjected to the following substitutions:

{\centering
     \texttt{.} $\mapsto$ \texttt{\$}, \quad\quad\quad
     \texttt{,} $\mapsto$ \texttt{\$\_\$} \quad\quad\quad
     \texttt{<} $\mapsto$ \texttt{\$\_\$}\quad\quad\quad
     \texttt{>} $\mapsto$ \texttt{\$\_\$}
\\}
 
These substitutions are essential, so that the class name satisfies the Java specifications.

\medskip
Differing from \cite{ORW00} we leave the class loader unchanged and 
for each used type $\funtypestar{\textit{N}}{ty_1 \ldo ty_n}{ty_0}$
we implement empty interface-files
%
  \verb;Fun;\texttt{\textit{N}}\verb;$$$_$;$ty_1$\verb;$_$; \ldots
  \verb;$_$;$ty_n$\verb;$_$;$ty_0$\verb;$_$;
%
 that, on the one hand, inherit
  from  
%
  \verb;Fun;\texttt{\textit{N}}\verb;$$$_$;$ty'_1$\verb;$_$; \ldots
  \verb;$_$;$ty'_n$\verb;$_$;$ty'_0$\verb;$_$;
%
if $\funtypestar{\textit{N}}{ty_1 \ldo ty_n}{ty_0}$ is a direct subtype from a
used type $\funtypestar{\textit{N}}{ty'_1 \ldo ty'_n}{ty'_0}$ 
and on the other hand, inherit form
\langelem{Fun1\$\$}. \langelem{Fun1\$\$} is
the original translation from any type $\funtypestar{\textit{N}}{ty_1 \ldo ty_n}{ty_0}$
which contains the \texttt{apply}-method.

  \begin{example}
    For the class \langelem{OLFun} in Example \ref{exa:OLFun} the following
    cutout of the bytecode is generated:
{\normalshape
\begin{lstlisting}[basicstyle=\ttfamily\footnotesize]
Double m(Fun1$$<Double,Double>);<@\red{descriptor:(LFun1\$\$\$\_\$Double\$\_\$Double\$\_\$;)Double;}@>
Integer m(Fun1$$<Integer,Integer>); <@\red{descriptor:(LFun1\$\$\$\_\$Integer\$\_\$Integer\$\_\$;)Integer;}@>
String m(Fun1$$<String,String>); <@\red{descriptor:(LFun1\$\$\$\_\$String\$\_\$String\$\_\$;)String;}@>
\end{lstlisting}
}
The ambiguities thus are resolved.
Note that the prefix of the primitive types
\langelem{(Ljava\$lang\$)} was left out for the sake of readability.
  \end{example}

\section{Related work}
\label{sec:relatedwork}
Some object-oriented languages such as Scala, C\# and Java perform
\emph{local} type inference \cite{PT98,OZZ01}. Local type 
inference means that missing type annotations are recovered only using 
information from adjacent nodes in the syntax tree without long distance
constraints. For instance, the type of a variable initialised with a
non-functional expression or the return type of a method can be
inferred. However, method argument types, in particular for recursive
methods, cannot be inferred by local type inference.

Milner's algorithm $\mathcal{W}$ \cite{Mil78,DM82} is
the gold standard for global type inference for languages with 
parametric polymorphism, which is used by ML-style languages. The fundamental idea
of the algorithm is to enforce type equality by means of many-sorted type
unification \cite{Rob65,MM82}. This approach is effective and results
in so-called principal types because many-sorted unification is
unitary, which means that there is at most one most general result.

The presence of subtyping means that type
unification is no longer unitary, but still finitary. Thus, there is
no longer a single most general type, but any type is an instance of a
finite set of maximal types. We have offered an algorithm for \java type
unification in \cite{plue09_1} and proved its soundness and completeness.

\pizza \cite{OW97} contains real function types with invariant arguments (no
subtyping). Function types similar  to those in \javatx are contained in \scala
\cite{scala34sortendauthorlist}. As shown in \cite{plue17_2}, our approach preserves the properties
of target typing, while \scala do not have this property.

Another feature derived from functional programming languages, which
has been introduced into \java in a restricted version, is pattern
matching. Pattern matching for the \texttt{in\-stance\-of} operator was introduced in
\java \textsf{13-16} \cite{JEP394} for the \texttt{switch-case} instruction
in \java \textsf{17-22} \cite{JEP427} and for the new record classes in \java
\textsf{19} \cite{JEP405}. 
In \pizza \cite{OW97}, pattern matching
is realised via algebraic data types for the switch-case statement.

There are different approaches for formal models of \java \cite{IPW01,TEP05,BBDGV18}. 
%
In \cite{SPT22} we presented an extension of the calculus in \cite{IPW01} with
type inference. This could be considered as a theoretical basis for \javatx.


\section{Summary and outlook}
Within the last 15 years, \java has been developed so as to introduce various concepts 
from functional programming languages.
In this paper we presented an extension of \java, called \javatx.
\javatx continues the range of incorporating functional programming
language features into \java. We added the feature of global
type inference. Global type inference means that \java programs can
be written without any type annotation. \javatx preserves static
typing. 

We presented a concept for generalisation for free type
variables (generated generics) which is more powerful than in functional
programming languages. The remaining type variable constraints
of the type inference were distributed to the class and its method,
respectively.

Subsequently, we showed how global type inference allows the extension of the
overloading mechanism such that not only method identifiers,
but also complete method declaration can be overloaded.


Another extension of \javatx is the introduction of real function
types. 
We introduced Scala-like function
types. For lambda expressions, we defined these function types as
explicit types. At the same time, we preserved the
concept of target typing for functional interfaces as was proposed in the
so-called strawman approach in order to introduce lambda expressions into
\java.
%
%
Furthermore, we achieved an approach for heterogeneous
translation of function types.

In combination with global type inference, an approach of pattern matching
in method headers similar to that of \haskell might be possible.  
We are working on this at the present time.
In summary then, we completed the features \emph{type inference} and
  \emph{lambda expressions} as well as the feature \emph{pattern matching}
  so that they are now as powerful in object-oriented languages as they are in functional programming languages.

\bibliographystyle{eptcs}
\bibliography{martin}
\end{document}